\documentclass[12pt,a4paper]{article}

\usepackage{amsfonts,amsmath,amssymb,amsthm}
\usepackage{latexsym,amscd}
\usepackage{amsbsy}
\usepackage{CJK}
\usepackage{indentfirst,mathrsfs}
\usepackage{amsfonts,amsmath,amssymb,amsthm}
\usepackage{latexsym,amscd}
\usepackage{amsbsy}
\usepackage{color}
\usepackage{multirow}
\usepackage{makecell}
\usepackage{float}
\usepackage{cases}

\numberwithin{equation}{section}

\newtheorem{thm}{Theorem}
\newtheorem{lem}{Lemma}

\newtheorem{remark}{Remark}

\newcommand{\fq}{{\mathbb F}_{q}}
\newcommand{\fqtwo}{{\mathbb F}_{q^2}}

\newcommand{\ftwon}{{\mathbb F}_{2^n}}

\newcommand{\ftwom}{{\mathbb F}_{2^m}}

\newcommand{\ftwo}{{\mathbb F}_{2}}
\newcommand{\Tr}{{\rm {Tr}}}

\newcommand{\ffc}{f_{\underline{c}}}

\def\oc{{\overline{c}}}
\def\oa{{\overline{a}}}
\def\oz{{\overline{z}}}
\def\ob{{\overline{b}}}
\def\os{{\overline{s}}}
\def\orr{{\overline{r}}}

\def\oxi{{\overline{\xi}}}

\def\ox{{\overline{x}}}
\def\ot{{\overline{\theta}}}
\def\te{{\theta}}

\begin{document}

\title{On permutation quadrinomials and $4$-uniform BCT}
\author{ Nian Li
\thanks{N. Li and X. Zeng are at the Hubei Key Laboratory of Applied Mathematics, Faculty of Mathematics and Statistics, Hubei University, Wuhan, 430062, China, and also with the State Key Laboratory of Cryptology, P.O. Box 5159, Beijing 100878, China.  Email: nian.li@hubu.edu.cn, xzeng@hubu.edu.cn}
, Maosheng Xiong
\thanks{M. Xiong is at the Department of Mathematics, The Hong Kong University of Science and Technology,
Clear Water Bay, Kowloon, Hong Kong, China. E-mail: mamsxiong@ust.hk}
and Xiangyong Zeng
}
\date{}
\maketitle

\begin{quote}
{\small {\bf Abstract:}
We study a class of general quadrinomials over the field of size $2^{2m}$ with odd $m$ and characterize conditions under which they are permutations with the best boomerang uniformity, a new and important parameter related to boomerang-style attacks. This vastly extends previous results from several recent papers.
}

\end{quote}

\section{Introduction} \label{sec1}

\subsection{Background}

In symmetric key cryptography, Substitution boxes (S-boxes) are basic components to perform substitutions. Being the only source of nonlinearity in many well-known block ciphers such as IDEA, AES and DES \cite{Knud}, they play a central role in obscuring the relationship between the key and ciphertext, the perplexity property depicted by Shannon \cite{Shannon}. The security of such ciphers depends crucially on the quality of the S-boxes used. It is thus important to find new designs of S-boxes with good cryptographic properties with respect to various attacks \cite{BSD,Lai,M,WAG}.

Mathematically, S-boxes are vectorial (multi-output) Boolean functions, that is, functions $F: V \to V'$ where $V$ and $V'$ are $m$ and $n$-dimensional vector spaces over the binary field $\ftwo$ respectively. 


Differential attack, proposed by Biham and Shamir \cite{BSD}, is one of the most fundamental cryptanalytic tools to assess the security of block ciphers. For an $n$-bit S-box $F:\ftwo^n \to \ftwo^n$, the properties for differential propagations of $F$ are captured in the DDT (Difference Distribution Table) of $F$ which are given by
$${\rm DDT}_{F}(a,b)=\#\left\{x\in \ftwo^n: F(x)+F(x+a)=b \right\} \quad \forall a, b \in \ftwo^n. $$
The differential uniformity of $F$ is defined as
$$\delta(F)=\max_{a,b\in \ftwo^n,\; a \ne \bf{0}}{\rm DDT}_{F}(a,b).$$
Differential uniformity is an important concept in cryptography as it quantifies the degree of security of the cipher with respect to differential attack if $F$ is used as an S-box in the cipher. In particular, if $\delta(F) =2$, then $F$ is called an almost perfect nonlinear (APN) function, which offers maximal resistance to differential attacks.

Boomerang attack is an important cryptanalysis technique introduced by Wagner \cite{WAG} in 1999 against block ciphers involving S-boxes. It can be considered as an extension of the classical differential attack \cite{BSD}. In a boomerang attack, the target cipher is regarded as a composition of two sub-ciphers, and two differentials are combined and analyzed for the upper and the lower parts of the cipher. The reader is referred to \cite{BDK01,BDK02,BDD03,BK09,DKS10,KKS01,KHP+12,Song-QH} for more details.

At Eurocrypt 2018, Cid, Huang, Peyrin, Sasaki and Song \cite{CHP} introduced a new tool called Boomerang Connectivity Table (BCT) to measure the resistance of a block cipher against the boomerang attack. The BCT can be used to more accurately evaluate the probability of generating a right quartet in boomerang-style attacks, and it provides more useful information when compared with the DDT \cite{CHP}. Let $F:\ftwo^n \to \ftwo^n$ be a permutation. The entries of the BCT of $F$ are given by
$${\rm BCT}_{F}(a,b)=\# \left\{x\in \ftwo^n: F^{-1}(F(x)+b)+F^{-1}(F(x+a)+b)=a \right\},$$
where $F^{-1}$ denotes the compositional inverse of $F$. The boomerang uniformity of $F$, introduced by Boura and Canteaut in \cite{BCO}, is defined as
$$\beta(F)=\max_{a,b\in \ftwo^n \setminus\{\bf{0}\}}{\rm BCT}_{F}(a,b).$$
The function $F$ is called a $\beta(F)$-uniform BCT function.

Roughly speaking, S-boxes $F$ with smaller value $\beta(F)$ provide stronger security against boomerang-style attacks. It was known in \cite{CHP} that $\beta(F) \ge \delta(F)$, and if $\delta(F)=2$, then $\beta(F)=2$, hence APN permutations offer maximal resistance to both differential and boomerang attacks. However, given the difficulty of finding APN permutations in even dimension (This is the Big APN Problem \cite{BDMW}), in even dimension which is the most interesting for real applications, we are contented with the next best, that is, permutations $F$ with $\beta(F)=4$.

Compared with an abundance of differentially $4$-uniform permutations in the literature (see \cite{Bracken-L,BTT,Gold,Kasami,Nyberg} for primary constructions and \cite{Peng-T,Qu-TLG,Tan-QTL,Tang-CT} and the references therein for constructions via the inverse function), it seems much harder to find permutations with $4$-uniform BCT in even dimension. Currently only six families of such permutations have been discovered (see \cite{BCO,LQSL,Li-Qu,Li-Xiong,MTX,TLZ} for details).


In particular, in \cite{TLZ} the authors studied a class of quadrinomial permutations of the form
\begin{eqnarray*}
  F(x)=x^{3q}+a_1x^{2q+1}+a_2x^{q+2}+a_3x^3\in \fqtwo[x],  \quad \forall\; a_1,a_2,a_3 \in \fqtwo
\end{eqnarray*}
where $q$ is an odd power of $2$, and derived general conditions on the coefficients $a_i$'s under which $F$ is a permutation and $\beta(F)=4$, and very recently, in \cite{Li-Qu} and independently in \cite{Li-Xiong} the authors considered the generalized butterfly structure (see \cite{FFW,LTYW, PUB}) and showed that the closed butterfly yields permutations with $4$-uniform BCT under suitable conditions. It was pointed out in \cite{Li-Qu,Li-Xiong} that the closed butterfly can be equivalently expressed as the univariate form
\begin{eqnarray} \label{1:fxx}
  c_0z^{(2^k+1)q}+c_1z^{2^kq+1}+c_2z^{q+2^k}+c_3z^{2^k+1},\; z\in \fqtwo
\end{eqnarray}
for some special $c_0, c_1, c_2, c_3\in\mathbb{F}_q$.

The objective of this paper is to study quadrinomials of the form \eqref{1:fxx} for much more general coefficients $c_i$'s and investigate conditions under which they become permutations with $4$-uniform BCT.

\subsection{Statement of the main result}

Throughout this paper, let $m$ and $k$ be both odd integers such that $\gcd(m,k)=1$. Let $n=2m$.  For any $x \in \ftwon$, denote $\overline{x}:=x^{2^m}$. For any $\underline{c}:=(c_0,c_1,c_2,c_3) \in \ftwon^4$, we consider a general quadrinomial $ f_{\underline{c}}: \ftwon \to \ftwon$ of the form
\begin{eqnarray}\label{1:fx}
  f_{\underline{c}}(x) &=& c_0\overline{x}^{2^k+1}+c_1\overline{x}^{2^k}x+c_2\overline{x}x^{2^k}+c_3x^{2^k+1}.
\end{eqnarray}
Denote
\begin{eqnarray*}
\begin{array}{llllll}
  \theta_1&=&c_0\overline{c}_0+c_1\overline{c}_1+c_2\overline{c}_2+c_3\overline{c}_3, &\theta_2&=&\overline{c}_0c_1+\overline{c}_2c_3, \\
 \theta_3&=&c_0\overline{c}_2+c_1\overline{c}_3,  &\theta_4&=&c_1\overline{c}_1+c_4\overline{c}_4,
\end{array}
\end{eqnarray*}
and define
\begin{eqnarray}\label{1:Gamma}
  \Gamma &=& \left\{\underline{c}\in\ftwon^{\;4}: \theta_1\ne 0, \Tr_1^m\left(\frac{\theta_4}{\theta_1}\right)=1, \left(\frac{\theta_2}{\theta_1}\right)^{2^k}=\frac{\overline{\theta}_3}{\theta_1}\right\}.
\end{eqnarray}
The set $\Gamma$ can be partitioned as $\Gamma=\Gamma_0 \cup \Gamma_1$, where
\begin{eqnarray}
\label{1:Gamma-i}
  \Gamma_i &=& \left\{\underline{c}\in\Gamma:   \Tr_1^m\left(\frac{\theta_2\overline{\theta}_2}{\theta_1^2}\right)=i\right\}, \quad i =0,1.
\end{eqnarray}
Our main result is stated as follows.
\begin{thm} \label{1:mainthm}
Let the setting be as above, and $f_{\underline{c}}$, $\Gamma$ and $\Gamma_i$ be defined by \eqref{1:fx}, \eqref{1:Gamma} and \eqref{1:Gamma-i} respectively.
\begin{enumerate}
\item[(1)] If $\underline{c}\in\Gamma$, then $f_{\underline{c}}$ is a permutation on $\ftwon$;
\item[(2)] If $\underline{c}\in\Gamma_0$, then $\beta(f_{\underline{c}})=\delta(\ffc)=4$;
\item[(3)] If $\underline{c}\in\Gamma_1$, then $\beta(f_{\underline{c}})\ge \delta(\ffc)=2^{m+1}$.
\end{enumerate}
\end{thm}

\begin{remark}
In the setting of Therorem \ref{1:mainthm}, if $k$ is even, $m$ is odd, $\gcd(m,k)=1$ and $\ffc$ is still of the form \eqref{1:fx}, then letting $k':=m-k$, we can obtain
\[\ffc(x)^{2^{k'}}= c_0' \ox^{2^{k'}+1}+c_1' \ox^{2^{k'}}x+c_2'\ox x^{2^{k'}}+c_3'x^{2^{k'}+1} , \]
where
\[c_0'=c_2^{2^{k'}}, c_1'=c_0^{2^{k'}}, c_2'=c_3^{2^{k'}}, c_3'=c_1^{2^{k'}}. \]
Noting that $k'$ is odd and $\gcd(m,k')=1$, denoting $\underline{c}':=(c_0',\ldots,c_3')$ and appealing to Theorem \ref{1:mainthm}, we can still obtain similar conditions to 1)-3) under which we can conclude that $\ffc$ is a permutation; $\beta(\ffc)=4$ and $\delta(\ffc)=2^{m+1}$. For the sake of simplicity, we omit the details.
\end{remark}

\begin{remark}
Similar to \cite{TLZA}, by using affine equivalence, the coefficients $c_i$'s of the quadrinomial $\ffc$ in \eqref{1:fx} may be simplified: if $c_0c_3 \ne 0$, we may assume that $c_0=1$; by considering $\ffc(\lambda x)$ for some $\lambda \in \ftwon^*$, we may assume that $c_1 \in \ftwom$. Actually when $k=1$, $c_0=1$ and $c_1\in\fq$, the function $f_{\underline{c}}(x)$ was originally studied in \cite{TLZ,TLZA,TZH}. In fact in this case 1) coinsides with the main result of \cite{TLZA} and 2) coinsides with the main result of \cite{TLZ}. On the other hand, using the special parametrization appearing in the papers, one can easily verify that \cite[Theorem 2]{Li-Xiong} and \cite[Theorem 1.1]{Li-Qu} can be derived from (1) and (2) of Theorem \ref{1:mainthm}. 
\end{remark}

\begin{remark}
Our computer experiments seem to indicate that if $f_{\underline{c}}$ is a permutation over $\ftwon$, then it is necessary that $\underline{c}\in\Gamma$. When $k=1$, this is indeed the case and was recently proved in \cite{Li-QLC}. For a general $k$, however, the method used there does not seem to work. We will come back to this question in the near future. If this ``necessity property'' were proved, then Theorem \ref{1:mainthm} indicates that $\ffc$ is a permutation with $4$-uniform BCT if and only if $\underline{c} \in \Gamma_0$, that is, the set $\Gamma_0$ completely charaterizes permutaitons $\ffc$ with $4$-uniform BCT. This may be another reason why we would expect that \cite[Theorem 2]{Li-Xiong} and \cite[Theorem 1.1]{Li-Qu} can be derived from (1) and (2) of Theorem \ref{1:mainthm}.
\end{remark}

\begin{remark}
Finally, for two permutations $F$ and $G$ over $\ftwon$, it is known that $\beta(F)=\beta(G)$ if $G=F^{-1}$ or $F$ and $G$ are affine equivalent \cite{BCO}; and if both $F$ and $G$ are quadratic and extended affine equivalent, then $\beta(G)=4$ if $\beta(F)=4$ \cite{MTX}. We have checked that for $m=3$, all the functions $\ffc$ for $\underline{c} \in \Gamma_0$ are affine equivalent to the Gold function $x^{2^k+1}$, which is known to be a permutation of $\mathbb{F}_{2^{2m}}$ with $4$-uniform BCT. It might be interesting to know if this holds for a general odd $m \ge 5$, or if there are permutations $\ffc$ with $4$-uniform BCT which are not affinely equivalent to the Gold function. In Table \ref{tab-known} we list all known permutations over $\ftwon$ with $\beta(F)=4$ for even $n$.
\end{remark}

\begin{table}\caption{Known permutations $F(x)$ with $\beta(F)=4$ over $\fqtwo$ for $q=2^m$}
\begin{center}\label{tab-known}
\begin{tabular}{|c|c|c|c|}
\hline
 No. & $F(x)$ & Reference & Remark\\
\hline
\hline
1 & $x^{-1}$ &  \cite{BCO} & $m$ odd \\
\hline
 2 &  $x^{2^t+1}$ &  \cite{BCO} & $m$ odd, $\gcd(2m,t)=2$\\
\hline
 3 &  $x^{q+2}+\gamma x$ &    \cite{LQSL}& Equivalent to No. 2 \\
\hline
 4 &  $\alpha x^{2^s+1}+\alpha^{2^{t}}x^{2^{-t}+2^{t+s}}$ &  \cite{MTX}    & - \\
\hline
 5 &  $x^{3q}+a_1x^{2q+1}+a_2x^{q+2}+a_3x^3$ &  \cite{TLZ} & Covered by No. 7\\
 \hline
 6 &  $c_0x^{(2^k+1)q}+c_1x^{2^kq+1}+c_2x^{q+2^k}+c_3x^{2^k+1}$ &  \cite{Li-Qu,Li-Xiong} & Covered by No. 7 \\
 \hline
 7 &  $c_0x^{(2^k+1)q}+c_1x^{2^kq+1}+c_2x^{q+2^k}+c_3x^{2^k+1}$ & This paper & Equivalence unknown \\
 \hline
\end{tabular}
\end{center}
\end{table}

The rest of this paper is organized as follows: in Section \ref{sec2:pre} we collect some solvability criteria on certain equations over finite fields which will be used repeatedly in the paper; in Section \ref{sec3:app} we present some identities and relations involving the $c_i$'s and $\theta_i$'s from the quarinomial $\ffc$; in Section \ref{sec4:diff} we discuss in details the solvability of the difference equation $\ffc(x+a)+\ffc(x)=b$; in Section \ref{sec5:mainthm} which is the longest section of the paper, we prove the main result, dealing with Parts (1), (3) and (2) of Theorem \ref{1:mainthm} individually in three seperate subsections.


\section{Preliminaries} \label{sec2:pre}


The following three results will be used repeatedly in the rest of the paper.

\begin{lem} (\cite{Lidl}) \label{lem0-0}
Let $n$ be a positive integer. For any $a \in \ftwon^*:=\ftwon \setminus \{0\}$ and $b \in \ftwon$, the equation
\[x^{2}+ax+b=0\]
is solvable (with two solutions) in $\ftwon$ if and only if
\[\Tr_1^n\left(\frac{b}{a^2}\right)=0.\]
Here $\Tr_{1}^n(\cdot)$ is the absolute trace map from $\mathbb{F}_{2^n}$ to the binary field $\mathbb{F}_{2}$.
\end{lem}

\begin{lem}  (\cite{Kim}) \label{lem0}
Let $n,k$ be positive integers such that $\gcd(n,k)=1$. For any $a \in \ftwon$, the equation
\[x^{2^k}+x=a\]
has either 0 or 2 solutions in $\ftwon$. Moreover, it is solvable with two solutions in $\ftwon$ if and only if $\Tr_1^n(a)=0$.
\end{lem}

\begin{lem} (\cite{Li-Xiong}) \label{lemma-core}
Let $m, k$ be odd integers such that $\gcd(k,m)=1$. Let $n=2m$. For any $\tau,\nu\in\ftwon$, define
\begin{equation*}
L_{\tau,\nu}(x)=x^{2^k}+\tau\overline{x}+(\tau+1)x+\nu.
\end{equation*}
Denote by $N(\tau,\nu)$ the number of solutions of  $L_{\tau,\nu}(x)=0$ in $\ftwon$. Then $N(\tau,\nu) \in \{0,2,4\}$. More precisely, let $\lambda \in\ftwom$ and $\Delta, \mu\in\ftwon$ be defined by the equations
\begin{eqnarray*}
\lambda^{2^k-1}=1+\tau+\overline{\tau},\;\;\; \Delta=\frac{\nu}{\lambda^{2^k}},\;\;\; \mu^{2^k}+\mu=\tau \lambda.
\end{eqnarray*}
Then
\begin{enumerate}
    \item [(1)] $N(\tau,\nu)=2$ if and only if one of the following conditions is satisfied:\\
         (i)  $1+\tau+\overline{\tau}=0$ and $\sum_{i=0}^{m-1}\left(\tau^{2^k}(\nu+\overline{\nu})+\nu^{2^k}\right)^{2^{ki}}=\nu+\overline{\nu}$;\\
         (ii) $1+\tau+\overline{\tau}\neq 0$, $\Tr_{1}^{n}(\Delta)=0$ and $\overline{\mu}+\mu=\lambda+1$.
    \item [(2)] $N(\tau,\nu)=4$ if and only if
     $1+\tau+\overline{\tau}\neq 0$, $\Tr_{1}^{n}(\Delta)=0$, $\overline{\mu}+\mu=\lambda$ and
     $\Tr_{1}^{n}\left(\frac{\mu^{2^k}\overline{\nu}}{\lambda^{2^k}}\right)=0$.
\end{enumerate}
If $\nu=0, 1+\tau+\overline{\tau} \ne 0$ and $\mu+\overline{\mu}=\lambda$, then the set of four solutions of $L_{\tau,\nu}(x)=0$ in $\ftwon$ is given by $\left\{0,1,\mu, \mu+1 \right\}$.
\end{lem}

\begin{remark} When $k=1$, Lemma \ref{lemma-core} reduces to \cite[Lemma 3]{TLZ} which played a central role in computing the boomerang uniformity in the paper. Comparing with \cite[Lemma 3]{TLZ}, our criteria seems a little simpler.
\end{remark}

\section{Some identities} \label{sec3:app}

Before proceeding to the proof of the main result, in this section we collect some useful identities and relations which play important roles in the rest of the paper.

Recall the setting of Theorem \ref{1:mainthm} in Section \ref{sec1} for all the notions $m,n,k, \Gamma$, etc.

\subsection{For $\underline{c} \in \Gamma$}

We first assume that $\underline{c}=(c_0,c_1,c_2,c_3) \in \Gamma$, and the function $\ffc$ is given in \eqref{1:fx}. Since  $\te_1,\te_4 \in \ftwom^*:=\ftwom \setminus \{0\}$ and
\[\Tr_1^m\left(\frac{\te_4}{\te_1}\right)=1, \quad \Tr_1^n\left(\frac{\te_4}{\te_1}\right)=0,\]
we can find $\xi \in \ftwon$ such that
\begin{eqnarray} \label{3:xixi} \xi^{2^k}+\xi=\frac{\te_4}{\te_1}. \end{eqnarray}
We fix such an element $\xi$. It is known that
\begin{eqnarray} \label{3:xir} \xi \in \ftwon \setminus \ftwom, \quad \xi+\overline{\xi}=1. \end{eqnarray}

\begin{lem} \label{2:ids} If $\underline{c} \in \Gamma$, then we have
\begin{enumerate}
  \item [(1)] $\te_2 \ot_2+\te_3 \ot_3=\te_4(\te_1+\te_4)$;
  \item [(2)] $c_0\te_4+c_1\ot_2+c_2\te_3=0$;
  \item [(3)] $c_0\te_2+c_1(\te_1+\te_4)+c_3\te_3=0$;
  \item [(4)] $c_0\ot_3+c_2(\te_1+\te_4)+c_3\ot_2=0$;
  \item [(5)] $c_1\ot_3+c_2\te_2+c_3\te_4=0$;
  \item [(6)] $\frac{\te_2 \ot_2}{\te_1^2}=\overline{\xi}+\xi^2+\Tr^m_1 \left(\frac{\te_2 \ot_2}{\te_1^2}\right)$.
\end{enumerate}
\end{lem}

\begin{proof}
Identities (1)-(5) can be verified by a routine computation. Only (6) requires some explanation.

Since $\te_1\in\ftwom^*$, we let $t=\frac{\te_2 \ot_2}{\te_1^2}$. Dividing $\te_1^2$ on both sides of Identity (1) of Lemma \ref{2:ids} and using the relation $\left(\te_2/\te_1\right)^{2^k}=\ot_3/\te_1$, we obtain
\[t^{2^k}+t+\frac{\te_4}{\te_1}+\left(\frac{\te_4}{\te_1}\right)^2=0. \]
Since $\te_1 \ne 0$ and $\gcd\left(2^k-1,2^n-1\right)=2^{\gcd(n,k)}-1=1$, by using $\xi$ from \eqref{3:xixi}, the above equation has two roots which are given by
\[t=\xi+\xi^2 \quad \mbox{ or } \quad t=\overline{\xi}+\xi^2. \]
It is easy to see that
\[\Tr_1^m \left(\xi+\xi^2\right)=\sum_{i=0}^{m-1}\xi^{2^i}+\sum_{i=0}^{m-1}\xi^{2^{i+1}}=\xi+\overline{\xi}=1, \]
and $\Tr_1^m \left(\overline{\xi}+\xi^2\right)=\Tr_1^m \left(1+\xi+\xi^2\right)=0$. Thus $t=\overline{\xi}+\xi^2$ if $\Tr_1^m(t)=0$ and  $t=\xi+\xi^2$ if $\Tr_1^m(t)=1$. This completes the proof of (6).
\end{proof}

\subsection{For $\underline{c} \in \Gamma_0$}

Next we assume that $\underline{c} \in \Gamma_0$. First, identity (6) of Lemma \ref{2:ids} becomes
\begin{eqnarray*}
\frac{\te_2 \ot_2}{\te_1^2}=\overline{\xi}+\xi^2. 
\end{eqnarray*}
Next, for any $a \in \ftwon^*$, define
\begin{eqnarray} \label{3:M}
M(a)&:=&\te_1 a \oa +\ot_2 \oa^2+\te_2 a^2,
\end{eqnarray}
and
\begin{eqnarray} \label{ap:eta}
\eta(a)&:=& \xi a+\frac{\ot_2 \oa}{\te_1}.
\end{eqnarray}
Define $\eta^{(2)}:=\eta \circ \eta$ (this is to avoid confusion which might result from using the more standard notation $\eta^2$). It is easy to see that
\begin{eqnarray*}
\eta^{(2)}(a)&:=& \overline{\xi} a+\frac{\ot_2 \oa}{\te_1}=\eta(a)+a, \\
\eta^{(3)}(a)&:=& \eta \circ \eta^{(2)} (a)=a. \nonumber
\end{eqnarray*}
Define
\begin{eqnarray} \label{ap:za}
Z_a&:=& \left\{a, \eta(a), \eta^{(2)}(a)\right\}.
\end{eqnarray}

\begin{lem} \label{3:Mnot0}
If $\underline{c} \in \Gamma_0$, then for any $a \in \ftwon^*$, we have
\begin{itemize}
\item[(1)] $M(a) \ne 0$;
\item[(2)] $M(z)=M(a)$ for any $z \in Z_a$.
\end{itemize}
\end{lem}
\begin{proof}
(1). Suppose $M(a)=0$ for some $a \in \ftwon^*$. Obviously $\te_2 \ne 0$. Let $x \in \ftwon^*$ be the unique element of $\ftwon$ satisfying $x^2=\te_2a^2$. Thus $a^2=x^2/\te_2$, and we obtain
\[0=M(a)^2=x^4+\ox^4+ \frac{\te_1^2}{\te_2 \ot_2} x^2 \ox^2. \]
Letting $y=\left(\ox/x\right)^2 \ne 0$, we have
\begin{eqnarray} \label{ap:lam1} y^2+\frac{\te_1^2}{\te_2 \ot_2} y+1=0.\end{eqnarray}
Since $\underline{c} \in \Gamma$, we have $\Tr_1^m \left(\frac{\te_2\ot_2}{\te_1^2}\right)=0$, Lemma \ref{lem0-0} implies that  \eqref{ap:lam1} is solvable with $y \in \ftwom^*$. From $y=\overline{y}$, we find $x^4=\ox^4$, that is, $x=\ox$, and hence $y=1$. This clearly contradicts   \eqref{ap:lam1} since we know that $\te_1\te_2 \ne 0$.

(2). Let $z=\eta(a)$. We have
\[z=\xi a +\frac{\ot_2 \oa}{\te_1}, \quad \oz=\overline{\xi} \oa +\frac{\te_2 a}{\te_1},\]
and
\begin{eqnarray*}
z \oz&=&\left(\xi a +\frac{\ot_2 \oa}{\te_1}\right)
\left(\overline{\xi} \oa +\frac{\te_2 a}{\te_1}\right) \\
&=&a \oa \xi \overline{\xi}+\frac{1}{\te_1} \left(\te_2 \xi a^2+\ot_2 \overline{\xi} \oa^2\right)+\frac{\te_2\ot_2 a \oa}{\te_1^2}.
\end{eqnarray*}
With some computation, we can obtain that
\begin{eqnarray*}
M(z)&=&a \oa \left(\te_1 \xi \overline{\xi}+\frac{\te_2\ot_2}{\te_1}\right)+\oa^2 \left(\ot_2 \oxi+\ot_2 \oxi^2+\frac{\te_2\ot_2^2}{\te_1^2}\right)+a^2 \left(\te_2 \xi+\te_2 \xi^2+\frac{\ot_2\te_2^2}{\te_1^2}\right).
\end{eqnarray*}
Using (6) of Lemma \ref{2:ids} and the properties of $\xi$ given in \eqref{3:xixi} and \eqref{3:xir}, we can verify that
\[\te_1 \xi \overline{\xi}+\frac{\te_2\ot_2}{\te_1}=\te_1\left(\xi \oxi+\oxi+\xi^2\right)=\te_1,\]
\[\te_2 \xi+\te_2 \xi^2+\frac{\ot_2\te_2^2}{\te_1^2}=\te_2 \left(\xi+\xi^2+\oxi+\xi^2\right)=\te_2.\]
This clearly shows that $M(z)=M(a)$.

Similarly, by taking $z=\eta^{(2)}(a)$, one can also verify that $M(z)=M(a)$. This completes the proof of (2). Now Lemma \ref{3:Mnot0} is proved.
\end{proof}

\begin{lem} \label{ap:idf}
If $\underline{c} \in \Gamma_0$, then for any $z \in \ftwon^*$, we have the identity
\begin{eqnarray} \label{ap:idf1}
\eta^{(2)}(z)^{2^k} \left(c_2 \overline{\eta(z)}+c_3 \eta(z)\right)+\overline{\eta^{(2)}(z)}^{2^k} \left(c_0 \overline{\eta(z)}+c_1 \eta(z)\right)=\ffc(z).
\end{eqnarray}
\end{lem}

\begin{proof}
For simplicity, denote
\[r:=\xi, \quad s:=\frac{\te_2}{\te_1}. \]
Then
\[\eta(z)=rz +\os \, \oz, \quad \eta^{(2)}(z)=\orr z+ \os \, \oz. \]
The left hand side of \eqref{ap:idf1} is given by
\begin{eqnarray*}
\mbox{LHS}&=&\left(\orr z+\os \, \oz\right)^{2^k} \left(c_2 \left(\orr \, \oz+ sz\right)+c_3 \left(r z+\os \, \oz\right) \right) \\
&& + \left(r \oz+s z\right)^{2^k} \left(c_0 \left(\orr \, \oz+ sz\right)+c_1 \left(r z+\os \, \oz\right) \right)\\
&=:& A \oz^{2^k+1}+B \oz^{2^k}z+C \oz z^{2^k}+D z^{2^k+1},
\end{eqnarray*}
where the coefficients $A,B,C,D$ are given by
\begin{eqnarray*}
A&=& \os^{2^k} \left(c_2 \orr+ c_3 \os\right)+ r^{2^k} \left(c_0 \orr+c_1 \os\right)=\orr\left(c_2 \os^{2^k}+ c_0 r^{2^k}\right)+ \os \left(c_3 \os^{2^k}+c_1 r^{2^k}\right), \\
B&=& \os^{2^k} \left(c_2 s+ c_3 r\right)+ r^{2^k} \left(c_0 s+c_1 r\right)=s \left(c_2 \os^{2^k}+ c_0 r^{2^k}\right)+ r \left(c_3\os^{2^k} +c_1 r^{2^k}\right), \\
C&=& \orr^{2^k} \left(c_2 \orr+ c_3 \os\right)+ s^{2^k} \left(c_0 \orr+c_1 \os\right)=\orr \left(c_2 \orr^{2^k}+ c_0 s^{2^k}\right)+ \os \left(c_3 \orr^{2^k}+c_1 s^{2^k}\right), \\
D&=& \orr^{2^k} \left(c_2 s+ c_3 r\right)+ s^{2^k} \left(c_0 s+c_1 r\right)=s \left(c_2 \orr^{2^k}+ c_0 s^{2^k}\right)+ r \left(c_3 \orr^{2^k}+c_1 s^{2^k}\right).
\end{eqnarray*}
We claim that $A=c_0,B=c_1,C=c_2$ and $D=c_3$. For $A$ and $B$, using the relations
\begin{eqnarray} r=\xi, \quad s=\frac{\te_2}{\te_1}, \quad \left(\frac{\ot_2}{\te_1}\right)^{2^k}=\frac{\te_3}{\te_1},\nonumber \\
\label{ap:xih} \xi^{2^k}+\xi=\frac{\te_4}{\te_1}, \quad \frac{\te_2 \ot_2}{\te_1^2}=1+\xi \oxi, \end{eqnarray}
and recalling (2)-(3) of Lemma \ref{2:ids}, we can obtain
\begin{eqnarray*}
c_2 \os^{2^k}+ c_0 r^{2^k}&=&c_2 \left(\frac{\te_3}{\te_1}\right)+c_0 \left(\frac{\te_4}{\te_1}+\xi\right)\\
&=& \frac{c_2 \te_3+c_0\te_4}{\te_1}+c_0 \xi=\frac{c_1 \ot_2}{\te_1}+c_0 \xi, \\
c_3 \os^{2^k}+c_1 r^{2^k}&=&c_3 \left(\frac{\te_3}{\te_1}\right)+c_1 \left(\frac{\te_4}{\te_1}+\xi\right)\\
&=& \frac{c_3 \te_3+c_1\left(\te_4+\te_1\right)}{\te_1}+c_1 \left(1+\xi\right)=\frac{c_0 \te_2}{\te_1}+c_1 \oxi.
\end{eqnarray*}
From the above identities and also using \eqref{ap:xih}, we can easily verify that
\begin{eqnarray*}
A&=& \oxi \left(\frac{c_1 \ot_2}{\te_1}+c_0 \xi\right)+\frac{\ot_2}{\te_1} \left(\frac{c_0 \te_2}{\te_1}+c_1 \oxi\right)=c_0, \\
B&=& \frac{\te_2}{\te_1} \left(\frac{c_1 \ot_2}{\te_1}+c_0 \xi\right)+\xi \left(\frac{c_0 \te_2}{\te_1}+c_1 \oxi\right)=c_1.
\end{eqnarray*}
As for $C$ and $D$, using (4)-(5) of Lemma \ref{2:ids}, we can obtain
\begin{eqnarray*}
c_2 \orr^{2^k}+ c_0 s^{2^k}&=&c_2 \left(\oxi+\frac{\te_4}{\te_1}\right)+c_0 \left(\frac{\ot_3}{\te_1}\right)\\
&=& \frac{c_2 \left(\te_1+\te_4\right)+c_0\ot_3}{\te_1}+c_2 \xi=\frac{c_3 \ot_2}{\te_1}+c_2 \xi, \\
c_3 \orr^{2^k}+c_1 s^{2^k}&=&c_3 \left(\oxi+\frac{\te_4}{\te_1}\right)+c_1 \left(\frac{\ot_3}{\te_1}\right)\\
&=& \frac{c_3 \te_4+c_1\ot_3}{\te_1}+c_3\oxi=\frac{c_2 \te_2}{\te_1}+c_3 \oxi.
\end{eqnarray*}
Then can easily verify that
\begin{eqnarray*}
C&=& \oxi \left(\frac{c_3 \ot_2}{\te_1}+c_2 \xi\right)+\frac{\ot_2}{\te_1} \left(\frac{c_2 \te_2}{\te_1}+c_3 \oxi\right)=c_2, \\
D&=& \frac{\te_2}{\te_1} \left(\frac{c_3 \ot_2}{\te_1}+c_2 \xi\right)+\xi \left(\frac{c_2 \te_2}{\te_1}+c_3 \oxi\right)=c_3.
\end{eqnarray*}
This shows that the left hand side of \eqref{ap:idf1} equals $\ffc(z)$. This completes the proof of Lemma \ref{ap:idf}.
\end{proof}



\section{Solving $\ffc(x+a)+\ffc(x)=b$} \label{sec4:diff}


Now to prove our main result, for any $\underline{c} \in \Gamma$, we first study for any $a\in\ftwon^*$, $b\in\ftwon$ the equation
\begin{equation}\label{3:equation0}
\ffc(x+a)+\ffc(x)=b.
\end{equation}
Here $\ffc(x)$ is given by \eqref{1:fx}. Denote
\begin{eqnarray} \label{3:ha}
H_a(x):=\ffc(x+a)+\ffc(x)+\ffc(a).
\end{eqnarray}
Since $\ffc(x)$ is a quadratic polynomial, we have
\begin{equation*}
H_a(x)=\tau_{1}'\overline{x}^{2^k}+\tau_{2}'x^{2^k}+\tau_{3}'\overline{x}+\tau_{4}'x,
\end{equation*}
where
\begin{eqnarray*}
\begin{array}{llllll}
  \tau_{1}'&=&c_{0}\overline{a}+c_{1}a, &\tau_{2}'&=&c_{2}\oa +c_{3}a, \\
 \tau_{3}'&=&c_{0}\oa^{2^k}+c_{2}a^{2^k},  &\tau_{4}'&=&c_{1}\oa^{2^k} + c_{3}a^{2^k}.
\end{array}
\end{eqnarray*}
Equation \eqref{3:equation0} becomes
\begin{eqnarray} \label{3:hab} H_a(x)=\ffc(a)+b.\end{eqnarray}
Substituting $x$ with $ax$, the above becomes
\begin{equation} \label{equation1}
\tau_{1}\overline{x}^{2^k}+\tau_{2}x^{2^k}+\tau_{3}\overline{x}+\tau_{4}x+\tau_{5}= 0
\end{equation}
where $\tau_{5}=\ffc(a)+b$ and other $\tau_i$'s are given by
\begin{eqnarray*}
\begin{array}{llllll}
  \tau_{1}&=&\left(c_{0}\oa+c_{1}a\right)\oa^{2^k}, &\tau_{2}&=&\left(c_{2}\oa +c_{3}a\right)a^{2^k}, \\
 \tau_{3}&=&\left(c_{0}\oa^{2^k}+c_{2}a^{2^k}\right)\overline{a},  &\tau_{4}&=&\left(c_{1}\oa^{2^k} + c_{3}a^{2^k}\right)a.
\end{array}
\end{eqnarray*}
Taking $2^m$-th power on both sides of \eqref{equation1} gives
\begin{eqnarray} \label{equation1-1}
\overline{\tau}_{1}x^{2^k}+\overline{\tau}_{2}\overline{x}^{2^k}+
\overline{\tau}_{3}x+\overline{\tau}_{4}\overline{x}+\overline{\tau}_{5}= 0,
\end{eqnarray}
then by $\overline{\tau}_{2}\cdot \eqref{equation1}+\tau_{1}\cdot \eqref{equation1-1}$ one has
\begin{equation} \label{equation2}
v_{1}x^{2^k}+v_{2}\overline{x}+v_{3}x+v_{4} = 0,
\end{equation}
 where
\begin{eqnarray} \label{3q:vis}
\begin{array}{llllll}
 v_{1}&=&\tau_{1}\overline{\tau}_{1}+\tau_{2}\overline{\tau}_{2}, &v_2&=&\tau_{1}\overline{\tau}_{4}+\overline{\tau}_{2}\tau_{3}, \\
 v_{3}&=&\tau_{1}\overline{\tau}_{3}+\overline{\tau}_{2}\tau_{4},  &v_4&=&\tau_{1}\overline{\tau}_{5}+\overline{\tau}_{2}\tau_{5}.
\end{array}
\end{eqnarray}
It is easy to verify that the $v_i$'s and $\tau_i$'s satisfy the following properties:
\begin{itemize}
\item[(1)] $v_1+v_2+v_3=0$;
\item[(2)] $\tau_1+\tau_2=\tau_3+\tau_4=\ffc(a)$;
\item[(3)] $v_{4}=v_{1}+\tau_{1}\overline{b}+\overline{\tau}_{2}b$.
\item[(4)] $\tau_{1}\overline{v}_{3}+\tau_{2}v_{2}+\tau_{3}v_{1}=
    \tau_{1}\overline{v}_{2}+\tau_{2}v_{3}+\tau_{4}v_{1}=\tau_{1}\overline{v}_{4}+\tau_{2}v_{4}+\tau_{5}v_{1}=0$.
\end{itemize}
By some straightforward computation, we can obtain the values of the $v_i$'s as follows:
\begin{eqnarray}
\label{3:v1}
v_1&=& \left(a \oa\right)^{2^k} M(a),\\
\label{3:v2}
v_2&=& \left(a \oa\right)^{2^k+1} \left(\te_4+\left(\oa/a\right)^{2^k} \te_3+ \left(\oa/a\right)\ot_2\right),\\
\label{3:v3}
v_3&=& \left(a \oa\right)^{2^k+1} \left(\te_1+\te_4+\left(\oa/a\right)^{2^k} \te_3+ \left(a/\oa\right)\te_2\right),\\
\label{3:v4}
v_4&=&v_1+ \oa^{2^k} \left(\left(\oc_2 a+\oc_3 \oa\right)b+\left(c_0 \oa+c_1 a\right)\ob\right),
\end{eqnarray}
where $M(a)$ is defined in \eqref{3:M}.

\subsection{Case 1: $v_1 \ne 0$}

If $v_1 \ne 0$, then we can write \eqref{equation2} as
\begin{equation} \label{differential-equation}
x^{2^k}+\frac{v_2}{v_1} \, \overline{x}+\left(1+\frac{v_2}{v_1}\right)x+\frac{\tau_{1}\overline{b}+
\overline{\tau}_{2}b}{v_{1}} +1=0.
\end{equation}
It turns out that if $v_1 \ne 0$, \eqref{equation1}  and  \eqref{differential-equation} are equivalent with each other.
\begin{lem} \label{lem-4}
If $v_{1}\neq 0$, then \eqref{equation1}  and  \eqref{differential-equation} have the same set of solutions in $\ftwon$.
\end{lem}

\begin{proof}
It suffices to show that \eqref{equation1} can be derived from \eqref{differential-equation}. Noting that $v_1=\overline{v}_1$, by using \eqref{equation2}, we obtain
\begin{equation*}
x^{2^k}=\frac{v_{2}\overline{x}+v_{3}x+v_{4}}{v_{1}} {\;\;\rm and\;\;} \overline{x}^{2^k}=\frac{\overline{v}_{2}x+\overline{v}_{3}\overline{x}+\overline{v}_{4}}{v_{1}}.
\end{equation*}
Then we can compute
\begin{eqnarray*}
&&\tau_{1}\overline{x}^{2^k}+\tau_{2}x^{2^k}+\tau_{3}\overline{x}+\tau_{4}x+\tau_{5}\\
&=&\frac{\tau_{1}}{v_{1}}(\overline{v}_{2}x+\overline{v}_{3}\overline{x}+\overline{v}_{4})
 +\frac{\tau_{2}}{v_{1}}(v_{2}\overline{x}+v_{3}x+v_{4})
 +(\tau_{3}\overline{x}+\tau_{4}x+\tau_{5})\\
&=&\frac{1}{v_{1}}\left[\overline{x}\left(\tau_{1}\overline{v}_{3}+\tau_{2}v_{2}+\tau_{3}v_{1}\right)
 +x(\tau_{1}\overline{v}_{2}+\tau_{2}v_{3}+\tau_{4}v_{1})
 +(\tau_{1}\overline{v}_{4}+\tau_{2}v_{4}+\tau_{5}v_{1})\right] \\
&=&0,
\end{eqnarray*}
which is \eqref{equation1}. This completes the proof.
\end{proof}

Recall from \eqref{3:xixi} that we have defined $\xi \in \ftwon \setminus \ftwom$ to be an element satisfying
\begin{eqnarray*} 
\xi^{2^k}+\xi=\frac{\te_4}{\te_1}.
\end{eqnarray*}
Now returning to \eqref{differential-equation} and comparing it with Lemma \ref{lemma-core},  we have
\begin{lem} \label{3:gen-eq}
Suppose $v_1 \ne 0$. For \eqref{differential-equation}, denote
\[\tau=\frac{v_2}{v_1}, \quad \nu=1+\frac{\tau_1 \ob+\overline{\tau}_2 b}{v_1}.  \]
Then we have
\begin{eqnarray}
\label{3:tau}
\tau&=& \frac{\te_1 a \oa}{M(a)}\left(\frac{\te_4}{\te_1}+\left(\frac{\ot_2 \oa}{\te_1 a}\right)+\left(\frac{\ot_2 \oa}{\te_1 a}\right)^{2^k}\right), \\
\label{3:nu}
\nu &=& 1+\frac{\oa^{2^k+1} \left(\oc_3 b+c_0\ob\right)+\oa^{2^k}a \left(\oc_2 b+c_1\ob\right)}{\left(a \oa\right)^{2^k}M(a)}.
\end{eqnarray}
Let $\lambda \in\ftwom$ and $\mu \in\ftwon$ be defined by the equations
\begin{eqnarray*} 
\lambda^{2^k-1}=1+\tau+\overline{\tau},\;\;\; \mu^{2^k}+\mu=\tau \lambda.
\end{eqnarray*}
Then
\begin{eqnarray} \label{3:lam}
\lambda&=&\frac{M(a)}{\te_1a \oa},
\end{eqnarray}
and $\mu$ can be taken as
\begin{eqnarray}
\label{3:mu}
\mu &=&\xi+ \frac{\ot_2\oa}{\te_1 a}.
\end{eqnarray}
Further, one has that
\begin{eqnarray*} 
1 +\tau+\overline{\tau} \ne 0, \quad \mbox{ and } \quad \mu+\overline{\mu}=\lambda,
\end{eqnarray*}
so by Lemma \ref{lemma-core}, \eqref{differential-equation} always has either 0 or 4 solutions for any $a \in \ftwon^*$ and any $b \in \ftwon$. Moreover, if $\nu=0$, then \eqref{differential-equation} always has four roots in $\ftwon$, which are given by
\[0, 1, \xi+ \frac{\ot_2\oa}{\te_1 a}, \overline{\xi}+\frac{\ot_2\oa}{\te_1 a}. \]
\end{lem}

\begin{proof}
All of the above facts can be verified easily with some careful computation. First, by \eqref{3:v1} and \eqref{3:v2} we have
\begin{align}
\tau=&\frac{v_{2}}{v_{1}}=\frac{\te_1 a \oa}{M(a)}\left(\frac{\te_4}{\te_1}+\left(\frac{\oa}{ a}\right)^{2^k} \frac{\te_3}{\te_1}+\frac{\ot_2 \oa}{\te_1 a}\right).  \nonumber
\end{align}
Then the value of $\tau$ in \eqref{3:tau} follows from the relation $\te_3/\te_1=\left(\ot_2/\te_1\right)^{2^k}$.

Second, using the value of $\tau$ in \eqref{3:tau} and noting that $\te_1,\te_4, M(a) \in \ftwom$, one gets
\begin{eqnarray*} 
1+\tau+\overline{\tau}
&=&1+\frac{\te_1 a \oa}{M(a)}\left(\left(\frac{\ot_2 \oa}{\te_1 a}+\frac{\te_2 a}{\te_1 \oa}\right)+\left(\frac{\ot_2 \oa}{\te_1 a}+\frac{\te_2 a}{\te_1 \oa}\right)^{2^k}\right) \nonumber \\
&=&1+ \frac{\te_1 a \oa}{M(a)}\left(\frac{M(a)}{\te_1 a \oa}+1+\left(\frac{M(a)}{\te_1 a \oa}+1\right)^{2^k}\right) \nonumber\\
&=&\left(\frac{M(a)}{\te_1 a \oa}\right)^{2^k-1}.
\end{eqnarray*}
Now \eqref{3:lam} is clear due to $\gcd(2^k-1,2^n-1)=1$. Third, from \eqref{3:tau} and \eqref{3:lam} one has
\begin{eqnarray*}
\tau \lambda&=&\frac{\te_4}{\te_1}+\left(\frac{\ot_2 \oa}{\te_1 a}\right)+\left(\frac{\ot_2 \oa}{\te_1 a}\right)^{2^k}.
\end{eqnarray*}
From the value $\xi$ given in \eqref{3:xixi}, it is clear that the value $\mu$ given by \eqref{3:mu} is a solution to the equation $\mu^{2^k}+\mu=\tau \lambda$. Using this value of $\mu$ and noting that $\xi+\overline{\xi}=1$, one can easily verify that $\mu+\overline{\mu}=\lambda$. This completes the proof of Lemma \ref{3:gen-eq}.
\end{proof}

\subsection{Case 2: $v_1 =0$}
If $v_1=0$, then we have:
\begin{lem} \label{geq:lem} Suppose $v_1=0$. Then
\begin{itemize}
\item[(1)] $\underline{c} \in \Gamma_1$, $\te_1\te_2\te_3\te_4 \ne 0$ and $\frac{\te_2\overline{\te}_2}{\te_1^2}=\xi+\xi^2$;
\item[(2)] $\frac{\overline{a}\overline{\te}_2}{a\te_1}=\xi$ or $\overline{\xi}$;
\item[(3)] $v_2=v_3=0$;
\item[(4)] If $\tau_1 \ne 0$, then
\begin{eqnarray}
\label{geq:t3t1}
\frac{\tau_3}{\tau_1} &=& \left(\frac{\te_1 a}{\ot_2 \overline{a}}\right)^{2^k-1}, \quad
\frac{\tau_2}{\tau_1} = \left(\frac{\te_2 a^2}{\ot_2 \overline{a}^2}\right)^{2^k}.
\end{eqnarray}
\end{itemize}
\end{lem}
\begin{proof}
(1). If $v_1=0$, then
\begin{eqnarray} \label{geq:ma} M(a)=\te_1 a \overline{a}+\ot_2 \overline{a}^2+\te_2 a^2=0.  \end{eqnarray}
By Lemma \ref{3:Mnot0}, it is necessary that $\underline{c} \in \Gamma_1$, that is, $\Tr_{1}^m(\te_2\overline{\te}_2/\te_1^2)=1$. It is obvious that $\te_1\te_2\te_3\te_4 \ne 0$. The desired expression on $\frac{\te_2 \ot_2}{\te_1^2}$ follows directly from (6) of Lemma \ref{2:ids}.

(2). Multiplying $\frac{\ot_2}{\te_1^2a^2}$ on both sides of \eqref{geq:ma}, we obtain
\begin{eqnarray*}
 \left(\frac{\overline{a}\overline{\te}_2}{a\te_1}\right)^2+\left(\frac{\overline{a}\overline{\te}_2}{a\te_1}\right)+\frac{\te_2\overline{\te}_2}{\te_1^2}=0.
\end{eqnarray*}
This proves (2) of Lemma \ref{geq:lem}.

(3). Suppose $\frac{\overline{a}\overline{\te}_2}{a\te_1}=\xi$, that is, $\frac{\overline{a}}{a}=\frac{\xi\te_1}{\overline{\te}_2}$. Using the relations $\xi^{2^k}+\xi=\te_4/\te_1$ and $\left(\te_2/\te_1\right)^{2^k}=\ot_3/\te_1$, it is easy to compute that 
\begin{eqnarray*}
  \te_4+\left(\frac{\overline{a}}{a}\right)^{2^k}\te_3+\left(\frac{\overline{a}}{a}\right)\overline{\te}_2 &=&  \te_4+\left(\frac{\xi\te_1}{\overline{\te}_2}\right)^{2^k}\te_3+\left(\frac{\xi\te_1}{\overline{\te}_2}\right)\overline{\te}_2\\
&=&
     \te_4+\left(\xi+\frac{\te_4}{\te_1}\right)\cdot\frac{\te_1}{\te_3}\cdot\te_3+\xi\te_1
     =0,
\end{eqnarray*}
that is, $v_2=0$. The other case that $\frac{\overline{a}\overline{\te}_2}{a\te_1}=\overline{\xi}$ can be handled in the same way. Now using the identity $v_1+v_2+v_3=0$, we conclude that $v_2=v_3=0$.

(4). We only consider the case $\frac{\overline{a}\overline{\te}_2}{a\te_1}=\xi$ since the other case can be handled in exactly the same manner. For simplicity, denoting $\gamma=\frac{\overline{a}}{a}$, we have $\gamma=\frac{\xi\te_1}{\overline{\te}_2}$ and
\begin{eqnarray*}
\frac{\tau_3}{\tau_1} &=& \frac{(c_0\overline{a}^{2^k}+c_2a^{2^k})\overline{a}}{(c_0\overline{a}+c_1a)\overline{a}^{2^k}}=\frac{1}{\gamma^{2^k-1}}\cdot \frac{c_0\gamma^{2^k}+c_2}{c_0\gamma+c_1}.
\end{eqnarray*}
Using $\xi^{2^k}+\xi=\te_4/\te_1$, $\left(\te_2/\te_1\right)^{2^k}=\ot_3/\te_1$ and (2) of Lemma \ref{2:ids}, we can obtain
\begin{eqnarray*}
\frac{c_0\gamma^{2^k}+c_2}{c_0\gamma+c_1}
&=& \frac{c_0\left(\frac{\xi\te_1}{\overline{\te}_2}\right)^{2^k}+c_2}{c_0\frac{\xi\te_1}{\overline{\te}_2}+c_1}
 =\frac{c_0\left(\xi+\frac{\te_4}{\te_1}\right)\cdot\frac{\te_1}{\te_3}+c_2}{c_0\frac{\xi\te_1}{\overline{\te}_2}+c_1}
 \\
&=& \frac{(c_0\te_4+c_2\te_3+c_0\te_1\xi)\overline{\te}_2}{(c_0\te_1\xi+c_1\overline{\te}_2)\te_3}
 =\frac{\overline{\te}_2}{\te_3}=
\left(\frac{\te_1}{\overline{\te}_2}\right)^{2^k-1}.
\end{eqnarray*}
This gives the desired expression of $\tau_3/\tau_1$.

As for $\tau_2/\tau_1$, using $\gamma=\frac{\overline{a}}{a}=\frac{\xi \te_1}{\ot_2}$, it is easy to see that  \begin{eqnarray*}
  \frac{\tau_2}{\tau_1}
  &=& \frac{(c_2\overline{a}+c_3a)a^{2^k}}{(c_0\overline{a}+c_1a)\overline{a}^{2^k}}
=\frac{1}{\gamma^{2^k}} \frac{(c_2\gamma+c_3)}{(c_0\gamma+c_1)}
=\frac{c_2\te_1\xi+c_3\overline{\te}_2}{\gamma^{2^k}\left(c_0\te_1\xi+c_1\overline{\te}_2\right)}. \end{eqnarray*}
A direct calculation gives
 \begin{eqnarray*}
   \left(c_0\te_1\xi+c_1\overline{\te}_2\right) \left(\overline{\xi}+\frac{\te_4}{\te_1}\right)\frac{\te_1}{\te_3}
   &=&\frac{\te_1}{\te_3}\left(c_0\te_1\xi\overline{\xi}+c_1\overline{\te}_2\overline{\xi}+c_0\te_4\xi+\frac{c_1\overline{\te}_2\te_4}{\te_1}\right) \\
     &=& \frac{\te_1}{\te_3}\left((c_0\te_4+c_1\overline{\te}_2)\xi+(c_0\te_2+c_1\te_1+c_1\te_4)\frac{\overline{\te}_2}{\te_1}\right) \\
      &=& c_2\te_1\xi+c_3\overline{\te}_2,
 \end{eqnarray*}
where we have used the relations $\overline{\xi}=\xi+1$, $\xi^2+\xi=\te_2\overline{\te}_2/\te_1^2$ and (2)-(3) of Lemma \ref{2:ids}.
Thus we have
\begin{eqnarray*}
  \frac{\tau_2}{\tau_1}
  &=& \frac{1}{\gamma^{2^k}}\left(\overline{\xi}+\frac{\te_4}{\te_1}\right)\frac{\te_1}{\te_3}. \end{eqnarray*}
Recalling again the relations $\gamma=\overline{a}/a=\frac{\xi \te_1}{\ot_2}$, $\overline{\xi}^{2^k}+\overline{\xi}=\te_4/\te_1$ and $\left(\te_2/\te_1\right)^{2^k}=\ot_3/\te_1$ gives the desired expression of $\tau_2/\tau_1$. Now the proof of Lemma \ref{geq:lem} is complete.
\end{proof}


\section{Proof of Theorem \ref{1:mainthm}} \label{sec5:mainthm}

Recall all the notations from the setting of Theorem \ref{1:mainthm}. In this section we prove Theorem \ref{1:mainthm}.

\subsection{$\underline{c} \in \Gamma \Longrightarrow \ffc$ is a permutation}\label{sec-pp}

Let $\underline{c} \in \Gamma$ be fixed. To prove that $\ffc$ is a permutation, we show that the equation
\begin{equation}\label{4:equation01}
\ffc(x+a)+\ffc(x)=0
\end{equation}
is not solvable in $\ftwon$ for any $a \in \ftwon^*$. Following arguments from the previous section, we consider two cases, that $v_1\ne 0$ and $v_1=0$.

Case 1: $v_1\ne 0$.

For this case \eqref{4:equation01} is equivalent to \eqref{differential-equation} with $b=0$, which can be written  as
\begin{eqnarray*} \label{4:eqper} x^{2^k}+\tau \overline{x}+(1+\tau)x+1=0,\end{eqnarray*}
where $\tau$ is defined in \eqref{3:tau}. According to Lemma \ref{3:gen-eq}, since $m$ is odd and $\overline{\mu}+\mu=\lambda$, we have
\[\Tr_1^n \left(\frac{\mu^{2^k}}{\lambda^{2^k}}\right)=\Tr_1^n \left(\frac{\mu}{\lambda}\right)=\Tr_1^m \left(\frac{\overline{\mu}+\mu}{\lambda}\right)=\Tr_1^m(1)=1,\]
hence by Lemma \ref{lemma-core},  \eqref{differential-equation} with $b=0$ and equivalently \eqref{4:equation01} is not solvable in $\ftwon$ for this $a \in \ftwon^*$.

Case 2: $v_1=0$.

In this case \eqref{4:equation01} is equivalent to  \eqref{equation1} with $b=0$. Here $\tau_5=\ffc(a)$. 

We first claim that $\ffc(a)\ne 0$. Otherwise, since $\tau_1+\tau_2=\ffc(a)=0$, we have $\tau_1=\tau_2$. Suppose $\tau_1 =0$, then $\tau_2=0$, letting $\gamma=\overline{a}/{a}$, we have
\[c_0 \gamma+c_1=0, \quad c_2 \gamma +c_3=0, \]
which implies that $c_0 \overline{c}_0=c_1 \overline{c}_1, c_2 \overline{c}_2=c_3 \overline{c}_3$, so $\te_1=0$, a contradiction. Now suppose $\tau_1 \ne 0$, then $\tau_2/\tau_1=1$, from (4) of Lemma \ref{geq:lem} we have
\[\frac{\te_2 a^2}{\overline{\te}_2 \overline{a}^2}=1 \Longrightarrow \left(\frac{\overline{a}}{a}\right)^2=\frac{\te_2}{\ot_2}. \]
This implies that
\[ \left(\frac{\overline{a}\ot_2}{a \te_1}\right)^2 =\frac{\te_2 \ot_2}{\te_1^2}=\xi+\xi^2. \]
However, this contradicts (2) of Lemma \ref{geq:lem}, which states that $\frac{\overline{a}\ot_2}{a \te_1}=\xi$ or $\overline{\xi}$.

Thus we have $\tau_5 =\ffc(a) \ne 0$. Since $v_1=v_2=v_3=0$, from the expressions of $v_i$'s in \eqref{3q:vis} and the relation $\tau_1+\tau_2=\tau_3+\tau_4=\ffc(a) \ne 0$, it is clear that $\tau_i\ne 0$ for any $i=1,2,3,4$.

Since $\tau_i\ne 0$ for any $i$ and $\tau_4/\tau_2=\overline{\tau}_3/\overline{\tau}_1$ due to $v_2=0$,  we can write 
\begin{eqnarray} \label{sol:deq}
 \overline{x}^{2^k}+\frac{\tau_2}{\tau_1}x^{2^k}+ \frac{\tau_3}{\tau_1}\overline{x}+\frac{\tau_4}{\tau_1}x+\frac{\tau_5}{\tau_1}
 =\left(\overline{x}^{2^k}+\frac{\tau_3}{\tau_1}\overline{x}\right)+\frac{\tau_2}{\tau_1}\left(x^{2^k}+\frac{\overline{\tau}_3}{\overline{\tau}_1}x\right)+\frac{\tau_5}{\tau_1}.
\end{eqnarray}
This implies that \eqref{equation1} is equivalent to the system of equations
 \begin{eqnarray}
 \label{5:eq1}
   x^{2^k}+\frac{\overline{\tau}_3}{\overline{\tau}_1}x+y &=&0, \\
 \label{5:eq2}
   \overline{y}+\frac{\tau_2}{\tau_1}y+\frac{\tau_5}{\tau_1}&=&0.
 \end{eqnarray}
We next show that for any $y \in \ftwon$ that satisfies \eqref{5:eq2},  \eqref{5:eq1} is not solvable for $x \in \ftwon$.

To this end, letting $\gamma=\frac{\overline{a}}{a}$, recalling (2) of Lemma \ref{geq:lem}, we may assume that
\[ \frac{\overline{a}\ot_2}{a \te_1}=\xi. \]
Then (4) of Lemma \ref{geq:lem} can be simplified as
\begin{eqnarray} \label{5-1:sim} \frac{\tau_3}{\tau_1}=\frac{1}{{\xi}^{2^k-1}}, \quad \frac{\tau_2}{\tau_1}=\left(\frac{\overline{\xi}}{\xi}\right)^{2^k}. \end{eqnarray}
Noting that $\tau_5=f_{\underline{c}}(a)=\tau_1+\tau_2$, by using new variables
\[ x_1:=x \overline{\xi}, \quad y_1:=\overline{\xi}^{2^k}y,\]
it is easy to see that the system of equations \eqref{5:eq1}-\eqref{5:eq2} is equivalent to
\begin{eqnarray}
 \label{5:eq3}
   x_1^{2^k}+x_1+y_1 &=&0, \\
 \label{5:eq4}
   \overline{y}_1+y_1+1&=&0.
 \end{eqnarray}
This system is clearly not solvable for $(x_1,y_1) \in \ftwon^2$ since for any $y_1 \in \ftwon$ that satisfies \eqref{5:eq4}, we have
\[\Tr^n_1(y_1)=\Tr_1^m(y_1+\overline{y}_1)=\Tr_1^m(1)=1, \]
so \eqref{5:eq3} is not solvable for $x_1 \in \ftwon$.

Combining Cases 1 and 2 we conclude that  for any $a\in\ftwon^*$,  \eqref{equation1} with $b=0$ has no solution for $x \in \ftwon$. Thus $\ffc$ is a permutation. This completes the proof of Part (1) of  Theorem \ref{1:mainthm}.

\subsection{$\underline{c} \in \Gamma_1 \Longrightarrow \delta(\ffc) =2^{m+1}$}

Let $\underline{c} \in \Gamma_1$ be fixed. Since $\ffc$ is a permutation, to find $\delta(\ffc)$, we fix an arbitrary  $a \in \ftwon^*$, and we study the largest possible number of solutions to the equation
\begin{equation}\label{5-2:equation01}
\ffc(x+a)+\ffc(x)=b,
\end{equation}
as $b$ runs through the set $\ftwon^*$. Similar to Subsection \ref{sec-pp}, we consider two cases, that $v_1\ne 0$ and $v_1=0$.

Case 1: $v_1\ne 0$.

This case is simple: since  \eqref{5-2:equation01} is equivalent to  \eqref{differential-equation}, and according to Lemma \ref{3:gen-eq}, it has either $0$ or $4$ solutions for any $b$.

Case 2: $v_1=0$.

This case is slightly more complicated. Equation \eqref{5-2:equation01} is equivalent to \eqref{equation1}, which is linear in the variable $x$. The largest number of solutions is achieved when $\tau_5=0$, so we choose $b=\ffc(a)$. Equation \eqref{equation1} becomes
\begin{eqnarray} \label{5-2:eq01}
\tau_1 \overline{x}^{2^k}+\tau_2 x^{2^k}+\tau_3 \overline{x}+\tau_4 x=0.
\end{eqnarray}
Since $v_1=v_2=v_3=0$, from the expressions of $v_i$'s in \eqref{3q:vis} and the relation $\tau_1+\tau_2=\tau_3+\tau_4=\ffc(a)=b \ne 0$, it is clear that $\tau_i\ne 0$ for any $i=1,2,3,4$. The expression \eqref{sol:deq} with $\tau_5=0$ implies that \eqref{5-2:eq01} is equivalent to the system of equations
 \begin{eqnarray}
 \label{5-2:eq1}
   x^{2^k}+\frac{\overline{\tau}_3}{\overline{\tau}_1}x+y &=&0, \\
 \label{5-2:eq2}
   \overline{y}+\frac{\tau_2}{\tau_1}y&=&0.
 \end{eqnarray}
Letting $\gamma=\frac{\overline{a}}{a}$ and recalling (2) of Lemma \ref{geq:lem}, we may assume that
\[ \frac{\overline{a}\ot_2}{a \te_1}=\xi. \]
Using \eqref{5-1:sim} and the new variables
\[ x_1:=x \overline{\xi}, \quad y_1:=\overline{\xi}^{2^k}y,\]
it is easy to see that the system of equations \eqref{5-2:eq1}-\eqref{5-2:eq2} is equivalent to
\begin{eqnarray}
 \label{5-2:eq3}
   x_1^{2^k}+x_1+y_1 &=&0, \\
 \label{5-2:eq4}
   \overline{y}_1+y_1&=&0.
 \end{eqnarray}
Clearly \eqref{5-2:eq4} is equivalent to $y \in \ftwom$, and for each such $y$, \eqref{5-2:eq3} has exactly two solutions for $x \in \ftwon$, that is, the system of equations \eqref{5-2:eq3}-\eqref{5-2:eq4} has $2^{m+1}$ solutions.

Combining Cases 1 and 2 we conclude that  for $\delta(\ffc)=2^{m+1}$. This completes the proof of Part (3) of  Theorem \ref{1:mainthm}.

\subsection{$\underline{c} \in \Gamma_0 \Longrightarrow \beta(\ffc)=4$}

Let $\underline{c} \in \Gamma_0$ be fixed. To compute the boomerang uniformity of $\ffc$, we need considerably more effort.

First, for any $a,b \in \ftwon^*$, we consider \eqref{5-2:equation01}. Lemma \ref{3:Mnot0} implies that $v_1 \ne 0$, then by Lemma \ref{3:gen-eq}, the equation has either $0$ or $4$ solutions, hence $\delta(\ffc) =4$.

Next, we recall a new formulation of the boomerang uniformity of $f(x)$ in \cite{LQSL}, which allows us to compute $\beta(f)$ conveniently without using the compositional inverse $f^{-1}$:
\begin{lem}(\cite{LQSL})\label{LQSL}
Let $f(x)$ be a permutation over $\ftwon$. Denote by $S_{f}(a,b)$ the number of solutions $(x,y)\in \ftwon^2$ of the equation system
\begin{numcases}{}
f(x+a)+f(y+a)=b, \nonumber \\
f(x)+f(y)=b. \nonumber
\end{numcases}
Then
\[\beta(f)=\max \left\{S_f(a,b): a,b\in \ftwon^*\right\}. \]
\end{lem}
Since $\beta(\ffc) \ge \delta(\ffc)=4$, to complete the proof of Part (2) of Theorem \ref{1:mainthm}, it suffices to show that $S_{\ffc}(a,b)\leq 4$ for any $a,b \in \ftwon^*$. Now for any fixed $a,b \in \ftwon^*$, the value $S_{\ffc}(a,b)$ is equal to the number of solutions $(x,y) \in \ftwon^2$ of the following equation system
\begin{numcases}{}
\ffc(x+a)+\ffc(x)+\ffc(y+a)+\ffc(y)=0, \label{boomerang-diff-1} \\
\ffc(x)+\ffc(y)=b. \label{boomerang-diff-2}
\end{numcases}
Since $b \ne 0$, obviously $x+y \ne 0$.

We first consider \eqref{boomerang-diff-1}. Using the function $H_a(x)$ defined in \eqref{3:ha} which is linear in both $a$ and $x$,  \eqref{boomerang-diff-1} can be rewritten as
\[H_a(x)+H_a(y)=H_a(x+y)=0. \]
Letting $z=x+y \in \ftwon^*$, tracing back to \eqref{3:hab} with the right hand being 0, the above equation has roots of the form $ax$ where $x$ satisfies the equation
\[x^{2^k}+ \tau \overline{x}+(1+\tau)x=0,\]
and $\tau$ is given by \eqref{3:tau} in Lemma \ref{3:gen-eq}. Using Lemma \ref{3:gen-eq} and Lemma \ref{lemma-core}, we conclude that $z=ax \in Z_a$ where the set $Z_a=\left\{a, \eta(a),\eta^{(2)}(a)\right\}$ which has been defined in \eqref{ap:za}, and we have
\begin{eqnarray} \label{4:sumfz} \sum_{z \in Z_a}f(z)=\ffc(a)+\ffc\left(\eta(a)\right)+\ffc\left(\eta^{(2)}(a)\right)=H_a\left(\eta(a)\right)=0.\end{eqnarray}

Next, we consider \eqref{boomerang-diff-2}. Using $y=x+z$, \eqref{boomerang-diff-2} becomes
\begin{eqnarray} \label{4:hz} H_z(x)=\ffc(z)+b. \end{eqnarray}
It is known from Lemma \ref{3:gen-eq} that the above equation has at most four solutions in $\ftwon$ for each $z \in Z_a$, so immediately we obtain $\beta(\ffc) \le 12$. To find the exact value of $\beta(\ffc)$, we need to consider more carefully the solvability of \eqref{4:hz} for $z \in Z_a$.

Using the equivalence between \eqref{3:hab} and \eqref{differential-equation} and applying Lemma \ref{3:gen-eq}, we conclude that for any $z\in Z_a$, \eqref{4:hz} is equivalent to
\begin{align}  \label{boomerang-2}
x^{2^k}+\tau_z\overline{x}+(1+\tau_z)x+\nu_z= 0,
\end{align}
where $\tau_z$ and $\nu_z$ are given by
\begin{eqnarray*}
  \tau_z &=& \frac{\te_1 z \oz}{M(z)}\left(\frac{\te_4}{\te_1}+\left(\frac{\ot_2 \oz}{\te_1 z}\right)+\left(\frac{\ot_2 \oz}{\te_1 z}\right)^{2^k}\right), \nonumber \\
   \nu_z &=& 1+\frac{\oz^{2^k+1} \left(\oc_3 b+c_0\ob\right)+\oz^{2^k}z \left(\oc_2 b+c_1\ob\right)}{\left(z \oz\right)^{2^k}M(z)}. \label{4:vz}
\end{eqnarray*}
Here $M(z)$ is given in \eqref{3:M}. Let us define
\begin{eqnarray*}
A(z)&:=& c_0 \oz^{2^k+1}+c_1 \oz^{2^k}z, \\
B(z)&:=& c_2 \oz z^{2^k}+c_3 z^{2^k}z.
\end{eqnarray*}
It is easy to see that $A(z)+B(z)=\ffc(z)$, and $\nu_z$ can be written as
\begin{eqnarray} \label{4:nuz}
   \nu_z &=& 1+\frac{b \, \overline{B(z)}+\ob \, A(z)}{\left(z \oz\right)^{2^k}M(z)}.
\end{eqnarray}
Further, for $\lambda_z$ satisfying $\lambda^{2^k-1}_z=1+\tau_z+\overline{\tau}_z$ and $\mu_z$ satisfying $\mu_z^{2^k}+\mu_z=\tau_z \lambda_z$, we have
\begin{eqnarray} \label{4:lamz} \lambda_z &=&\frac{M(z)}{\te_1 z \oz},\\
\label{4:muzlam} \mu_z &=&\xi+\frac{\ot_2 \oz}{\te_1 z},\end{eqnarray}
where $\xi$ is defined in \eqref{3:xixi}.

Thus, by Lemma \ref{3:gen-eq} and Lemma \ref{lemma-core}, \eqref{boomerang-2} (and equivalently \eqref{4:hz}) has either $0$ or $4$ solutions in $\ftwon$, and it has $4$ solutions in $\ftwon$ if and only if
\begin{eqnarray} \label{4:4roots} \Tr_{1}^{n}\left(\frac{\nu_z}{\lambda^{2^k}_z}\right)=0 \;\;{\rm and}\;\;\Tr_{1}^{n}\left(\frac{\mu^{2^k}_z\overline{\nu}_z}{\lambda^{2^k}_z}\right)=0.\end{eqnarray}
According to \eqref{4:nuz} and \eqref{4:lamz}, it can be readily verified that
\begin{eqnarray} 
\Tr_{1}^{n}\left(\frac{\nu_z}{\lambda^{2^k}_z}\right)
&=&\Tr_{1}^{n}\left(\frac{1}{\lambda_z^{2k}}\right)+\Tr_{1}^{n}\left(\frac{\te_1^{2^k}\left(b \overline{B(z)}+\ob A(z)\right)}{M(z)^{2^k+1}}\right). \nonumber
\end{eqnarray}
Noting that $\lambda_z,\te_1, M(z) \in \ftwom$ and $A(z)+B(z)=\ffc(z)$, by using Lemma \ref{3:Mnot0}, for $z \in Z_a$, we obtain
\begin{eqnarray*} 
\Tr_{1}^{n}\left(\frac{\nu_z}{\lambda^{2^k}_z}\right)
&=&\Tr_{1}^{n}\left(\frac{\te_1^{2^k}\ob \left(A(z)+B(z)\right)}{M(z)^{2^k+1}}\right) =\Tr_{1}^{n}\left(\frac{\te_1^{2^k}\ob \ffc(z)}{M(a)^{2^k+1}}\right). \nonumber\\
\end{eqnarray*}
Equation \eqref{4:sumfz} indicates that $\sum_{z\in Z_a}\Tr_{1}^{n}\left(\frac{\nu_z}{\lambda^{2^k}_z}\right)=0$, hence $\Tr_{1}^{n}\left(\frac{\nu_z}{\lambda^{2^k}_z}\right)=0$ for exactly one $z\in Z_a$ or for all $z\in Z_a$. If $\Tr_{1}^{n}\left(\frac{\nu_z}{\lambda^{2^k}_z}\right)=0$ for exactly one $z\in Z_a$, then $S_{\ffc}(a,b)\leq 4$ by Lemma \ref{lemma-core}.

Now let us assume that
\begin{eqnarray} \label{4:dz0} \Tr_{1}^{n}\left(\frac{\nu_z}{\lambda^{2^k}_z}\right)=\Tr_{1}^{n}\left(\frac{\te_1^{2^k}\ob \ffc(z)}{M(a)^{2^k+1}}\right)=0 \quad \forall\; z\in Z_a.\end{eqnarray}

Using values $\nu_z$, $\lambda_z$ and $\mu_z$ given in \eqref{4:nuz}, \eqref{4:lamz} and \eqref{4:muzlam} respectively and noting that $\mu_z+\overline{\mu}_z=\lambda_z$, we obtain
\begin{eqnarray} \label{4:V}
T_z:= \Tr_{1}^{n}\left(\frac{\mu^{2^k}_z\overline{\nu}_z}{\lambda^{2^k}_z}\right)
&=&
\Tr_{1}^{n}\left(\frac{\mu^{2^k}_z}{\lambda^{2^k}_z} \left(1+\frac{\ob \, B(z)+b \, \overline{A(z)}}{\left(z \oz\right)^{2^k}M(z)}\right)\right) \nonumber \\
&=&
\Tr_{1}^{n}\left(\frac{\mu^{2^k}_z}{\lambda^{2^k}_z}\right)+\Tr_{1}^{n} \left(\frac{\mu^{2^k}_z \left(\ob \, B(z)+b \, \overline{A(z)}\right)}{\lambda^{2^k}_z\left(z \oz\right)^{2^k}M(z)}\right) \nonumber \\
&=&
1+\Tr_{1}^{n} \left(\frac{\te_1^{2^k} \, \ob}{M(z)^{2^k+1}} \left(\mu^{2^k}_z B(z)+\overline{\mu}_z^{2^k}A(z)\right)\right). \nonumber
\end{eqnarray}
Using $\eta$ given in \eqref{ap:eta} and noting that $\eta(z)=z \mu_z$
we have
\begin{eqnarray*}
T_z&=&1+\Tr_{1}^{n} \left(\frac{\te_1^{2^k} \, \ob}{M(a)^{2^k+1}} \left(\eta(z)^{2^k} \left(c_2 \oz+c_3 z\right)+\overline{\eta(z)}^{2^k} \left(c_0 \oz+c_1 z\right)\right) \right).
\end{eqnarray*}
Since it is known that $Z_a=\left\{a, \eta(a), \eta^{(2)}(a)\right\}$, $\eta^{(2)}(a)=\eta \circ \eta(a)=a+\eta(a)$ and $\eta \circ \eta^{(2)}(a)=a$, from Lemma \ref{ap:idf}, we can obtain the identity
\[\eta(z)^{2^k} \left(c_2 \oz+c_3 z\right)+\overline{\eta(z)}^{2^k} \left(c_0 \oz+c_1 z\right)=\ffc \left(\eta^{(2)}(z)\right), \]
that is,
\begin{eqnarray*}
T_z&=&1+\Tr_{1}^{n} \left(\frac{\te_1^{2^k} \, \ob \ffc\left(\eta^{(2)}(z)\right)}{M(a)^{2^k+1}} \right), \quad \forall \;z \in Z_a.
\end{eqnarray*}
Since $\eta^{(2)}(z) \in Z_a$, from \eqref{4:dz0} we derive that $T_z=1$ for any $z \in Z_a$. This means that \eqref{4:4roots} never holds for $z \in Z_a$, that is, $S_{\ffc}(a,b)=0$. Combining these two cases we conclude that $S_{\ffc}(a,b)\leq 4$ for any $a,b \in \ftwon^*$. Hence $\beta(\ffc) \le 4$. This completes the proof of Part (2) of Theorem \ref{1:mainthm}.

\section*{Acknowledgements}

 This work was supported by the National Natural Science Foundation of China (Nos. 61702166, 61761166010) and by the Research Grants Council (RGC) of Hong Kong (No. N\_HKUST169/17).



\begin{thebibliography}{99}


\bibitem{BDK01} E. Biham, O. Dunkelman, N. Keller, The rectangle attack-rectangling the Serpent, In Birgit Pfitzmann, editor, EUROCRYPT 2001,
LNCS, vol. 2045, pp. 340-357. Springer, Heidelberg, May 2001.

\bibitem{BDK02} E. Biham, O. Dunkelman, N. Keller, New results on boomerang and
rectangle attacks, In Joan Daemen and Vincent Rijmen, editors, FSE 2002,
LNCS, vol. 2365, pp. 1-16. Springer, Heidelberg, February 2002.

\bibitem{BSD} E. Biham, A. Shamir, Differential cryptanalysis of DES-like cryptosystems, J. Cryptology, 4(1) (1991), pp. 3-72.

\bibitem{BDD03} A. Biryukov, C. De Canni\`ere, G. Dellkrantz, Cryptanalysis
of SAFER++, In Dan Boneh, editor, CRYPTO 2003, LNCS, vol. 2729,
pp. 195-211. Springer, Heidelberg, August 2003.

\bibitem{BK09} A. Biryukov, D. Khovratovich, Related-key cryptanalysis of the
full AES-192 and AES-256, In Mitsuru Matsui, editor, ASIACRYPT 2009,
LNCS, vol. 5912, pp. 1-18. Springer, Heidelberg, December 2009.

\bibitem{BCO} C. Boura, A. Canteaut, On the boomerang uniformity of cryptographic sboxes, IACR Trans. Symmetric
Cryptol. 3 (2018), pp. 290-310.

\bibitem{Bracken-L} C. Bracken, G. Leander, A highly nonlinear differentially 4 uniform power mapping that permutes fields of even degree, Finite Fields Appl. 16 (2010), pp. 231-242.

\bibitem{BTT} C. Bracken, C. Tan, Y. Tan, Binomial differentially $4$-uniform permutations with high nonlinearity, Finite Fields Appl. 18(3) (2012), pp. 537-546.

\bibitem{BDMW} K.A. Browning, J.F. Dillon, M.T. McQuistan, A.J. Wolfe, An APN permutation in dimension six, Finite Fields Appl. 518 (2010), pp. 33-42.




\bibitem{CHP} C. Cid, T. Huang, T. Peyrin, Y. Sasaki, L. Song, Boomerang Connectivity Table: A new cryptanalysis tool,
 In Jesper Buus Nielsen and Vincent Rijmen, editors, Advances in Cryptology-EUROCRYPT 2018, pp. 683-714, Cham,
2018. Springer International Publishing.

\bibitem{DKS10} O. Dunkelman, N. Keller, A. Shamir, A practical-time related-key
attack on the KASUMI cryptosystem used in GSM and 3G telephony, In Tal
Rabin, editor, CRYPTO 2010, LNCS, vol. 6223, pp. 393-410. Springer,
--Heidelberg, August 2010.

\bibitem{FFW} S. Fu, X. Feng, B. Wu, Differentially $4$-uniform permutations with the best known nonlinearity from butterflies, IACR Trans. Symmetric Cryptol. (2) (2017), pp. 228-249.

\bibitem{Gold} R. Gold, Maximal recursive sequences with 3-valued recursive cross-correlation functions (corresp.), IEEE Trans. Inf. Theory 14(1)(1968), pp. 154-156.
    

\bibitem{Kasami} T. Kasami, The weight enumerators for several classes of subcodes of the 2nd order binary
reed-muller codes, Inf. Control. 18(4)(1971), pp. 369-394.

\bibitem{KKS01} J. Kelsey, T. Kohno, B. Schneier, Amplified boomerang
attacks against reduced-round MARS and Serpent, In Gerhard GoosJuris HartmanisJan van LeeuwenBruce Schneier, editors, FSE 2000. LNCS, vol. 1978, pp. 75-93. Springer, Berlin, Heidelberg.

\bibitem{KHP+12} J. Kim, S. Hong, B. Preneel, E. Biham, O. Dunkelman,
N. Keller, Related-key boomerang and rectangle attacks: Theory and
experimental analysis, IEEE Trans. Inf. Theory 58(7) (2012), pp. 4948-4966.


\bibitem{Kim}  K. Kim, J. Choe, D. Lee, D. Go, S. Mesnager. Solutions of $x^{q^k}+\cdots+x^q+x=a$ in $\mathbb{F}_{2^n}$, https://arxiv.org/pdf/1905.10579v1.pdf

\bibitem{Knud} L. R. Knudsen, M. J. B. Robshaw, The Block Cipher Companion, Springer, Berlin, 2011.

\bibitem{Lai}X. Lai, Higher order derivatives and differential cryptanalysis, Communications and Cryptography 276 (1994), pp. 227-233.

\bibitem{LQSL} K. Li, L. Qu, B. Sun, C. Li, New results about the
boomerang uniformity of permutation polynomials, IEEE Trans. Inf. Theory 65(11) (2019),
pp. 7542-7553.

\bibitem{Li-Qu} K. Li, C. Li, T. Helleseth, L. Qu, Cryptographically strong permutations from the butterfly structure, https://arxiv.org/abs/1912.02640

\bibitem{Li-QLC} K. Li, L. Qu, C. Li, H. Chen, On a conjecture about a class of permutation quadrinomials, https://arxiv.org/abs/1909.08209

\bibitem{Li-Xiong} N. Li, Z. Hu, M. Xiong, X. Zeng, 4-uniform BCT permutations from generalized butterfly structure. Under review.

\bibitem{LTYW} Y. Li, S. Tian, Y. Yu, M. Wang, On the generalization of butterfly structure, IACR Trans. Symmetric Cryptol. 2018(2) (2018), pp. 160-179.

\bibitem{Lidl} R. Lidl, H. Niederreiter, Finite Fields, Encyclopedia of Mathematics, vol. 20, Cambridge University Press, Cambridge, 1997.

\bibitem{M} M. Matsui, Linear cryptanalysis method for DES cipher, In Tor Helleseth, editor, Advances in Cryptology-EUROCRYPT'93, pp. 55-64, Berlin, Heidelberg, 1994. 




\bibitem{MTX} S. Mesnager, C. Tang, M. Xiong, On the boomerang uniformity of quadratic permutations, https://eprint.iacr.org/2019/277.pdf



\bibitem{Nyberg} K. Nyberg, Differentially uniform mappings for cryptography, In Tor Helleseth, editor, Advances in Cryptology-EUROCRYPT'93, pp. 134-144, Berlin, Heidelberg, 1994. 

\bibitem{Peng-T} J. Peng, C. Tan, New differentially 4-uniform permutations by modifying the inverse function on subfields, Cryptogr. Commun. 9 (2017), pp. 363-378.

\bibitem{PUB} L. Perrin, A. Udovenko, A. Biryukov, Cryptanalysis of a Theorem: Decomposing the only known solution to the big APN problem, In Matthew Robshaw, Jonathan Katz, editors, LNCS, vol. 9816, pp. 93-122. Springer, 2016.

\bibitem{Qu-TLG} L. Qu, Y. Tan, C. Li, G. Gong, More constructions of differentially 4-uniform permutations on $\mathbb{F}_{2^{2k}}$, Des. Codes Cryptogr. 78 (2) (2016), pp. 391-408.

\bibitem{Shannon} C. E. Shannon, Communication theory of secrecy systems, Bell Labs Technical Journal, vol. 28, no. 4, pp. 656--715, 1949.

\bibitem{Song-QH} L. Song, X. Qi, L. Hu. Boomerang connectivity table revisited: Application
to SKINNY and AES, https://eprint.iacr.org/2019/146.pdf

\bibitem{Tan-QTL} Y. Tan, L. Qu, C.H. Tan, C. Li, New families of differentially 4-uniform permutations over $\mathbb{F}_{2^{2k}}$, In Tor Helleseth and Jonathan Jedwab, editors, SETA 2012, LNCS, vol. 7280, pp. 25-39, Springer, 2012.

\bibitem{Tang-CT} D. Tang, C. Carlet, X. Tang, Differentially 4-uniform bijections by permuting the inverse function, Des. Codes. Cryptogr. 77(1)(2015), pp. 117-141.

\bibitem{TLZ} Z. Tu, N. Li, X. Zeng, J. Zhou, A class of quadrinomial permutation with boomerang uniformity four, to appear in IEEE Trans. Inf. Theory.

\bibitem{TLZA} Z. Tu, X. Liu, X. Zeng, A revisit of a class of permutation quadrinomial, Finite Fields Appl. 59 (2019), pp. 57-85.

\bibitem{TZH} Z. Tu, X. Zeng, T. Helleseth, New permutation quadrinomials over $\mathbb{F}_{2^{2m}}$, Finite Fields Appl. 50 (2018), pp. 304-318.


\bibitem{WAG} D. Wagner, The boomerang attack, In Lars R. Knudsen, editor, FSE'1999, LNCS, vol. 1636, pp. 156-170. Springer, Heidelberg, 1999.

\end{thebibliography}
\end{document}